\documentclass[journal,twoside]{IEEEtran}

\usepackage{cite}
\usepackage{mdwmath}
\usepackage{epsfig}
\usepackage{easymat}
\usepackage{amsmath}
\usepackage{float}
\usepackage{cite}
\usepackage{graphicx}
\usepackage{balance}

\newtheorem{theorem}{Theorem}
\newtheorem{corollary}{Corollary}
\newtheorem{Lemma}{Lemma}

\begin{document}

\title{Impact of Spatial Correlation on the Finite-SNR Diversity-Multiplexing Tradeoff}

\author{Z. Rezki,~\IEEEmembership{Student Member,~IEEE,}
        David Haccoun,~\IEEEmembership{Life-Fellow,~IEEE,}
        Fran\c{c}ois Gagnon,~\IEEEmembership{Senior Member,~IEEE,}
        and Wessam Ajib,~\IEEEmembership{Member,~IEEE,}%
\thanks{Z. Rezki and D. Haccoun are with the Department
of Electrical Engineering, \'Ecole Polytechnique de Montr\'eal,
Email: \{zouheir.rezki,david.haccoun\}@polymtl.ca,}
\thanks{F. Gagnon is with the Department of Electrical Engineering, \'Ecole de technologie
sup\'erieure, Email: francois.gagnon@etsmtl.ca}
\thanks{W. Ajib is with the Department of Computer Sciences, Universit\'e du Qu\'ebec \`a
Montr\'eal Email: ajib.wessam@uqam.ca}
\thanks{Manuscript received November 06, 2006; revised April 20, 2007; accepted June 04, 2007.
 The editor coordinating the review of this paper and approving it for publication was Jitendra
 Tugnait.}
\thanks{Digital Object Identifier 10.1109/TWC.2008.060984}}

\markboth{IEEE TRANSACTION ON WIRELESS COMMUNICATIONS,~Vol.~X, No.~X, XXX~2007}%
{Rezki \MakeLowercase{\textit{et al.}}: Impact of Spatial
Correlation on the Finite-SNR Diversity-Multiplexing Tradeoff}
%



\maketitle

\begin{abstract}
The impact of spatial correlation on the performance limits of
multielement antenna (MEA) channels is analyzed in terms of the
diversity-multiplexing tradeoff (DMT) at finite signal-to-noise
ratio (SNR) values. A lower bound on the outage probability is
first derived. Using this bound accurate finite-SNR estimate of
the DMT is then derived. This estimate allows to gain insight on
the impact of spatial correlation on the DMT at finite SNR. As
expected, the DMT is severely degraded as the spatial correlation
increases. Moreover, using asymptotic analysis, we show that our
framework encompasses well-known results concerning the asymptotic
behavior of the DMT.
\end{abstract}


\section{Introduction}\label{introduction}
Multielement antenna (MEA) systems have been used either to
increase the diversity gain in order to better combat channel
fading \cite{Tarokh98}, or to increase the data rate by means of
spatial multiplexing gain \cite{Foschini96}. Recently, Zheng and
Tse showed that both gains can be achieved with an asymptotic
optimal tradeoff at high-SNR regime \cite{Zheng03}. This tradeoff
is a characterization of the maximum diversity gain that can be
achieved at each multiplexing gain. At low to moderate SNR values
(typically $3-20$ dB), this asymptotic tradeoff is in fact an
optimistic upper bound on the finite-SNR diversity-multiplexing
tradeoff.  A framework was introduced by Narasimhan to
characterize the diversity performance of rate-adaptive MIMO
systems at finite SNR \cite{Nara}. As expected, the achievable
diversity gains at realistic SNR values are significantly lower
than for the asymptotic values. In \cite{Zheng03} and \cite{Nara},
the authors assume independent and identically distributed
(i.i.d.) fading channels. However, in real propagation
environments, the fades can be correlated which may be detrimental
to the performances of MEA systems \cite{Da}. Recently, the impact
of spatial correlation on the finite-SNR DMT was studied in
\cite{Nara_IT}.\

In this paper, we analyze the impact of spatial correlation on the
DMT at finite-SNR values using a new framework and we extend to
the high-SNR regimes, the results obtained for finite-SNR in
\cite{Rezki-vtc06}. We prove that our framework encompasses the
well-known asymptotic DMT for correlated and uncorrelated channels
\cite{Zheng03,Chang}. Thus, our framework may be seen as a
generalization of the asymptotic analysis.\

When capacity-achieving codes are used over a quasi-static
channel, the errors are mainly caused by atypical deep fades of
the channel, that is, the block-error probability is equal to the
outage probability of the channel $P_{out}$ which will be formally
defined later in the paper. Therefore, $P_{out}$ is the key
parameter for deriving the performance limits of MEA systems.
Since derivation of an exact expression for $P_{out}$ is
difficult, we alternatively derive lower bounds on $P_{out}$ over
both spatially correlated and uncorrelated channels. These bounds
are then used to obtain insightful estimates of the related
finite-SNR DMT. These estimates allow to characterize the
potential limits of MEA systems, in terms of DMT, in a more
realistic propagation environment and for practical SNR values.
The paper is organized as follows. Section \ref{model} presents
the system model and introduces the related definitions. In
section \ref{outage}, we derive lower bounds on $P_{out}$.
Finite-SNR DMT estimates are given in section \ref{diversity}. An
asymptotic analysis of the diversity estimates is investigated in
section \ref{sec-asymptotic}. Numerical results are reported in
Section \ref{result} and Section \ref{conclusion} concludes the
paper.

\section{Channel Model and Related definitions }\label{model}
Let an MEA system consist of $N_t$ transmit antennas and $N_r$
receive antennas. We restrict our analysis to a Rayleigh
flat-fading channel, where the entries of the $N_r \times N_t$
channel matrix ${\boldsymbol H}$ are circularly-symmetric zero
mean complex Gaussian distributed and possibly correlated. The
channel is assumed to be quasi-static, unknown at the transmitter
and completely tracked at the receiver. For the effect of spatial
fading correlation, we model ${\boldsymbol H}$ as:
$\boldsymbol{H}=\boldsymbol{R_{r}^{1/2}}\boldsymbol{H_w}\boldsymbol{R_{t}^{1/2}}$,
where matrix $\boldsymbol H_w$ represents the $N_r \times N_t$
spatially uncorrelated channel, and where matrices $\boldsymbol
R_{r}$, of dimension $N_{r} \times N_{r}$ and $\boldsymbol R_{t}$,
of dimension $N_{t} \times N_{t}$ are positive-definite Hermitian
matrices that specify the receive and transmit correlations
respectively \cite{Paulraj}.

For an SNR-dependent spectral efficiency $R(SNR)$ in bps/Hz,
$P_{out}$ is defined as \cite{Telatar}:
\begin{equation}\label{Outage}
    P_{out}=Prob(I<R)
\end{equation}
where $I$ is the channel mutual information. Assuming equal power
allocation over the transmit antennas, $I$ is given by:
\begin{equation}\label{I}
    I=\log_2{\left(\det{\left(\boldsymbol{I_{N_r}}+\frac{\eta}{N_t}{\boldsymbol
H}{\boldsymbol H}^H \right)}\right)} \quad bps/Hz
\end{equation}
where $\boldsymbol{I_{N_r}}$ is the $N_r \times N_r$ identity
matrix and where the superscript $^H$ indicates for conjugate
transposition. The multiplexing and diversity gains are
respectively defined as \cite{Nara}:
\begin{IEEEeqnarray}{rCl}\label{r}
     {r} & = & \frac{R}{\log_{2}(1+g \cdot \eta)}\\
     \label{d}
     d(r,\eta) & = & -\eta {\frac{\partial{\ln P_{out}(r,\eta)}}{\partial{
     \eta}}}
\end{IEEEeqnarray}
where $g$ is the array gain which is equal to
$N_r$, and where $\eta$ is the mean SNR value at each receive
antenna. As defined in (\ref{r}) and (\ref{d}), the multiplexing
gain $r$ provides an indication of a rate adaptation strategy as
the SNR changes, while the diversity gain $d(r,\eta)$ can be used
to estimate the additional SNR required to decrease $P_{out}$ by a
specific amount for a given $r$.

\section{Lower Bound on The Outage Probability}\label{outage}
Let us denote the orthogonal-triangular (QR) decomposition of
$\boldsymbol H_w=\boldsymbol Q \boldsymbol R$, where $\boldsymbol
Q$ is an $N_{r} \times N_{r}$ unitary matrix and where
$\boldsymbol R$ is an $N_{r} \times N_{t}$ upper triangular matrix
with independent entries. The square magnitudes of the diagonal
entries of $\boldsymbol R$, $|R_{l,l}|^2$, are chi-square
distributed with $2(Nr-l+1)$ degrees of freedom, $l=1,\ldots
,\min{(N_{t},N_{r})}$. The off-diagonal elements of $\boldsymbol
R$ are i.i.d. Gaussian variables, with zero mean and unit
variance. Let the singular value decompositions (SVDs) of
$\boldsymbol R_{t}^{1/2}=\boldsymbol U \boldsymbol D \boldsymbol
V^H$, and $\boldsymbol R_{r}^{1/2}=\boldsymbol U' \boldsymbol D'
\boldsymbol V'^H$, where $\boldsymbol U$ and $\boldsymbol V$ (or
$\boldsymbol U'$, $\boldsymbol V'$) are $N_{t} \times N_{t}$ (or
$N_{r} \times N_{r})$ that satisfy $\boldsymbol U ^H \boldsymbol U
=
 \boldsymbol V^H \boldsymbol V = \boldsymbol I_{N_t}$ (or $\boldsymbol U'^H \boldsymbol U'=
\boldsymbol V'^H \boldsymbol V' = \boldsymbol I_{N_r}$), and where
$\boldsymbol D$ (or $\boldsymbol D'$) is an $N_{t} \times N_{t}$
(or $N_{r} \times N_{r}$) diagonal matrix, whose diagonal elements
are the singular values of $\boldsymbol R_{t}^{1/2}$ (or
$\boldsymbol R_{r}^{1/2}$). We assume, without loss of generality,
that the elements of $\boldsymbol D$ and $\boldsymbol D'$ are
ordered in descending order of their magnitudes along the
diagonal. Using these SVDs, since
$\det{\left(\boldsymbol{I}+\boldsymbol{X}\boldsymbol{Y}\right)}=\det{\left(\boldsymbol{I}+\boldsymbol{Y}\boldsymbol{X}\right)}$
and since unitary transformations do not change the statistics of
random matrices, we have:
\begin{IEEEeqnarray}{rCl}
\nonumber
  I & = & \log_2{\left(\det{(\boldsymbol{I_{N_r}}+\frac{\eta}{N_t}{\boldsymbol
R_{r}^{1/2}} {\boldsymbol H_{w}}{\boldsymbol R_{t}}{\boldsymbol
H_{w}^{H}}{\boldsymbol R_{r}^{H/2}})}\right)}\\
    & \stackrel{d}{=} & \log_2{\left(\det{(\boldsymbol{I_{N_r}}+\frac{\eta}{N_t}{\boldsymbol
R} {\boldsymbol D^2} {\boldsymbol R^H} {\boldsymbol D^{'2}})
}\right)}  \label{UB-mutu-inf}
\end{IEEEeqnarray}

where the symbol $\stackrel{d}{=}$ means equality in
distributions. From (\ref{UB-mutu-inf}), it is clear that
$\boldsymbol R_{t}$ and $\boldsymbol R_{r}$ contribute to the
channel mutual information through their diagonal matrix
representatives ${\boldsymbol D^2}$ and ${\boldsymbol D^{'2}}$
respectively. Since ${\boldsymbol D^2}$ and ${\boldsymbol D^{'2}}$
have a similar role in (\ref{UB-mutu-inf}), without loss of
generality, we can focus on the spatial correlation at the
transmitter, that is, we assume $\boldsymbol{
R_{r}^{1/2}}={\boldsymbol D^{'2}}=\boldsymbol{I_{N_r}}$.\

Let $D_{k}$, $k=1,\ldots,N_{t}$, denote the $k^{th}$ diagonal
element of $\boldsymbol{D}$, and let $R_{l,k}$ represent the
element of $\boldsymbol R$ at the $l^{th}$ row and the $k^{th}$
column, $l=1,\ldots,N_{r}$, $k=1,\ldots,N_{t}$. Using the fact
that $\det(\boldsymbol{A})\leq \prod_{l}A_{l,l}$, for any
nonnegative-definite matrix $\boldsymbol{A}$, we obtain from
(\ref{UB-mutu-inf}):
\begin{equation}\label{I_upper}
   I \leq \sum_{l=1}^{t}\log_2{\left(1+\frac{\eta}{N_t} \Delta_l\right)},
\end{equation}
where $t=\min(N_t,N_r)$ and where
$\Delta_l=\sum_{k=l}^{N_t}D_{k}^2|{R_{l,k}}^2|$, $l=1,\ldots,t$,
is the $l^{th}$ diagonal entry of $\boldsymbol{R} \boldsymbol
{D^2} \boldsymbol {R^H}$. Since $R_{l,k}$ are independent, then
$\Delta_l$ are also independent. In order to derive a lower bound
on $P_{out}$, the distribution function of $\Delta_l$ is needed.
When all $D_{k}^2, k=1,..,N_t$ are equal which corresponds to the
uncorrelated case, the trace constraint $trace
(\boldsymbol{R_{t}^{1/2}})=N_t$ imposes $D_{k}^2=1$,
$k=1,\dots,N_t$. That is, $\Delta_l$ is chi-square distributed
with $2(N_{r}+N_{t}-2l+1)$ degrees of freedom. Otherwise,
$\Delta_l$ may be viewed as a generalized quadratic form of a
Gaussian random vector. We first derive the distribution function
of $\Delta_{l}$, $l=1,\dots,t$, in Lemma \ref{L1}.\\

\begin{Lemma}
\label{L1}
 Assuming that all $D_{k}^2$'s, $k=1,\dots,N_t$, are
distinct, the distribution function of $\Delta_{l}$, $l=1,\dots,t$,
is given by:
\begin{equation}\label{dist-delta}
      f_{\Delta_{l}}(x)=\sum_{k=1}^{N_{r}-l+1}
    a_{k}^{(l)}f_{G(k,D_{l}^2)}(x)+\sum_{k=1}^{N_{t}-l}a_{1}^{(l+k)}f_{G(1,D_{l+k}^2)}(x),
\end{equation}
where $G(\alpha,\beta)$ is a Gamma random variable with
probability distribution function given by:
$f_{G(\alpha,\beta)}(x)=\frac{x^{\alpha -1}}{\Gamma(\alpha)
\beta^{\alpha}} e^{-\frac{x}{\beta}}$, $x\geq 0$, $\alpha > 0$,
$\beta > 0$. The coefficients $a_{k}^{(l)}$ and $a_{1}^{(l+k)}$
are given by:
\begin{IEEEeqnarray*}{rcl}
   a_{k}^{(l)} & = &  \frac{(-D_{l}^2)^{-(N_{r}-l+1-k)}}{(N_{r}-l+1-k)!} \cdot\\
   & \quad &  \frac{d^{(N_{r}-l+1-k)}}{d(j
v)^{(N_{r}-l+1-k)}}\left[\left(1-jvD_{l}^2\right)^{N_{r}-l+1}\Psi_{\Delta_{l}}(jv)\right]_{jv=D_{l}^{-2}},\\
a_{1}^{(l+k)} & = &
  \left[\left(1-jvD_{l+k}^2\right)\Psi_{\Delta_{l}}(jv)\right]_{jv=D_{l+k}^{-2}},
\end{IEEEeqnarray*}
where $\frac{d^{(k)f(x)}}{dx^{(k)}}$ is the $k^{th}$ derivative of
$f(x)$ and where $\Psi_{\Delta_{l}}(jv)$ is the moment generating
function of $\Delta_{l}$ given by:
\begin{equation*}\label{}
    \Psi_{\Delta_{l}}(jv)=\left(1-jvD_{l}^{2}\right)^{-(N_{r}-l+1)}\prod_{k=1}^{N_{t}-l}
\left(1-jvD_{l+k}^{2}\right)^{-1}.
\end{equation*}
\end{Lemma}

\begin{proof}
The proof can be found in \cite{Rezki-vtc06}.
\end{proof}\

Note that when all $D_{k}^2, k=1,\dots,N_t$, are not distinct, the
distribution of $\Delta_{l}, l=1,..,t$, can be derived using the
same mechanism. Using (\ref{Outage}), (\ref{I_upper}) and Lemma
\ref{L1}, a lower bound on $P_{out}$ may be expressed in the
following theorem.\\

\begin{theorem}[Lower Bound]\label{th-LB}
\label{T1} Lower bounds on the outage probability $P_{out}$ for
the uncorrelated $D_{k}^2=1$, $k=1,\dots,N_t$, and correlated
spatial fading channels are respectively given by:
\begin{IEEEeqnarray}{rCl}
\label{eq-low-bound}
  P_{out}^{uncorr} & \geq &  \prod_{l=1}^{t}\Gamma_{inc}(\xi_l,N_{r}+N_{t}-2l+1), \\
\label{eq-low-bound2}
  P_{out}^{corr}   & \geq & \prod_{l=1}^{t} \Bigl( \sum_{k=1}^{N_{r}-l+1}a_{k}^{(l)}\Gamma_{inc}(\frac{\xi_l}{D_{l}^{2}},k)\\
   &\quad & +\sum_{k=1}^{N_{t}-l}a_{1}^{(l+k)}\Gamma_{inc}(\frac{\xi_l}{D_{l+k}^{2}},1)\Bigr),
   \nonumber
\end{IEEEeqnarray}
where $b_{l}$, $l=1,\dots,t$, are arbitrary positive coefficients
that satisfy $r=\sum_{l=1}^{t}b_l$, $\Gamma_{inc}$ is the
incomplete Gamma function defined by
$\Gamma_{inc}(x,a)=\frac{1}{(a-1)!}\int_{0}^{x}t^{a-1}e^{-t}dt$
and where $\xi_l$ is given by:
$\xi_l=\frac{N_t}{\eta}\left(\left(1+g\eta\right)^{b_l}-1\right)$.
\end{theorem}

\begin{proof}
The proof has been given in \cite{Rezki-vtc06} but follows along
similar lines as \cite[Theorem 1]{Nara_IT}.
\end{proof}\

In order to obtain tighter results, the lower bounds given in
Theorem \ref{T1} are maximized over the set of coefficients
$b_{l}$, $l=1,\dots,t$, for each multiplexing gain $r$ and each
SNR value $\eta$. Clearly, the computational time of this
optimization problem is much smaller than that required by Monte
Carlo simulations for computing the exact $P_{out}$. It is worth
noting that since $\xi_{l} \geq \alpha_{\rho,l}$, where
$\alpha_{\rho,l}$ was defined in \cite[Theorem 1]{Nara_IT}, and
since $\Gamma_{inc}(x,a)$ is an increasing function in $x$, the
lower bounds given by (\ref{eq-low-bound}) is tighter than that
given by \cite[Theorem 1]{Nara_IT} for the uncorrelated case.
Moreover, (\ref{eq-low-bound2}) appears as a finite product of a
weighted sum of $\Gamma_{inc}$ functions, which is more insightful
and easier to compute than the lower bound derived in
\cite[Theorem 1]{Nara_IT}, which involves weighted infinite series
of $\Gamma_{inc}$ functions.

\section{Finite-SNR Diversity and Correlation}\label{diversity}
Using Theorem \ref{th-LB} and (\ref{d}), an estimate of the
finite-SNR diversity for a given multiplexing gain is now derived
in the following corollary.\\

\begin{corollary}[Diversity estimate] \label{C1}
An estimate of the diversity for the correlated and uncorrelated
spatial fadings are respectively given by:
\begin{IEEEeqnarray}{rcl}
 \label{d-est-uncorr}
 \nonumber
 \hat{d}^{uncorr}(r,\eta)& = & \frac{N_t}{\eta} \sum _{l=1}^{t}\left((1+g\eta)^{b_l}-b_l g \eta (1+g\eta)^{b_l
 -1}-1\right)\\
  & \quad & \frac{\xi_{l}^{N_t + N_r -2l}e^{-\xi_l}/(N_t+N_r
 -2l)!}{{\Gamma_{inc}(\xi_l,N_r + N_t -2l+1)}}\\
 \label{d-est-corr}
 \nonumber
\hat{d}^{corr}(r,\eta)& = & \frac{N_t}{\eta} \sum_{l=1}^{t}
\left((1+g\eta)^{b_l}-b_l g \eta (1+g\eta)^{b_l
 -1}-1\right)\\
   & \quad & \frac{Q_{l}(\xi_{l})}{P_{l}(\xi_{l})},
 \end{IEEEeqnarray}
where $Q_{l}(\xi_{l})$ and $P_{l}(\xi_{l})$ are given by:
\begin{IEEEeqnarray}{rCl}
\nonumber
\label{Ql}
  Q_{l}(\xi_{l}) & = & \sum_{k=1}^{N_{r}-l+1}\frac{a_{k}^{(l)}}{(k-1)!}\left(\frac{\xi_l}{D_{l}^{2}}\right)^{k-1}e^{{-\frac{\xi_l}{D_{l}^{2}}}}D_{l}^{-2}\\
\nonumber
  & \quad & +\sum_{k=1}^{N_{t}-l}a_{1}^{(l+k)}e^{-\frac{\xi_l}{D_{l+k}^{2}}}D_{l+k}^{-2} \\
\nonumber
 \label{Pl}
  P_{l}(\xi_{l}) & = &
  \sum_{k=1}^{N_{r}-l+1}a_{k}^{(l)}\Gamma_{inc}\left(\frac{\xi_l}{D_{l}^{2}},k\right)\\
  \nonumber
  & \quad & +\sum_{k=1}^{N_{t}-l}a_{1}^{(l+k)}\Gamma_{inc}\left(\frac{\xi_l}{D_{l+k}^{2}},1\right).
\end{IEEEeqnarray}
\end{corollary}\

Note that (\ref{d-est-uncorr}) and (\ref{d-est-corr}) have similar
closed forms. Clearly, (\ref{d-est-uncorr}) can be obtained from
(\ref{d-est-corr}) by replacing $P_{l}(\xi_{l})$ and
$Q_{l}(\xi_{l})$ by ${\Gamma_{inc}(\xi_l,N_r + N_t -2l+1)}$ and
$\left(\xi_{l}^{N_t + N_r -2l}e^{-\xi_l}\right)/(N_t+N_r -2l)!$
respectively. It should be pointed out that (\ref{d-est-uncorr})
and (\ref{d-est-corr}) are simpler and more insightful than the
diversity estimates given in \cite[Theorem 3]{Nara_IT}, which
again involve infinite series.

\section{Asymptotic behavior of the diversity estimates}\label{sec-asymptotic}
In order to examine whether the diversity estimates given in
Corollary \ref{C1} match the well-known asymptotic DMT at high-SNR
given in \cite{Zheng03}, we analyze the asymptotic behavior of the
diversity estimates we derived, as $\eta \rightarrow \infty$ or as
$r \rightarrow 0$. First, we present the following lemma.\\

\begin{Lemma}
\label{lemmma-asymptoic-behavior} Assuming full-rank transmit
spatial correlation, we can write:
\begin{equation}\label{equal-d}
    \lim_{\eta \rightarrow \infty}\hat{d}^{uncorr}(r,\eta)=\lim_{\eta \rightarrow
    \infty}\hat{d}^{corr}(r,\eta).
\end{equation}
\end{Lemma}

\begin{proof}
For convenience, the proof is presented in Appendix \ref{app-1}.
\end{proof}\

The result in Lemma \ref{lemmma-asymptoic-behavior} is very
insightful. It states that, at a high-SNR regime and at a given
multiplexing gain $r$, the diversity estimate is independent of
the spatial correlation. More interestingly, our asymptotic
diversity estimates coincide with the well-known asymptotic DMT
characterization as summarized in the following theorem.\\

\begin{theorem}\label{th-equality-asymp}
Assuming full-rank transmit spatial correlation, the optimal DMT,
for the uncorrelated and correlated cases, is given by the
asymptotic diversity estimate, that is:
\begin{equation}\label{eq-equality-asymp}
    \lim_{\eta \rightarrow \infty}\hat{d}^{uncorr}(r,\eta)=\lim_{\eta \rightarrow
    \infty}\hat{d}^{corr}(r,\eta)=d_{asym},
\end{equation}
where $d_{asym}=-\lim_{\eta \rightarrow
\infty}\frac{\log_{2}P_{out}}{\log_{2}\eta}$.
\end{theorem}

\begin{proof}
The proof is presented in Appendix \ref{app-2}.
\end{proof}\

Theorem \ref{th-equality-asymp} states that our asymptotic
diversity estimate is exactly the high-SNR DMT. Therefore, it can
be seen as a generalization of the DMT for the spatially
correlated and uncorrelated channels. Note that Lemma
\ref{lemmma-asymptoic-behavior} and Theorem
\ref{th-equality-asymp} agree with the results in \cite[Corollary
1]{Nara_IT}. On the other hand, Lemma
\ref{lemmma-asymptoic-behavior} and Theorem
\ref{th-equality-asymp} confirm a recently established result
concerning  the asymptotic diversity \cite{Chang}. However, our
result is broader, since it allows understanding the impact of
spatial correlation at finite-SNR which is not discussed in
\cite{Chang}. More importantly, the framework presented here
provides some guidelines on designing space-time codes at
practical SNR values. As an example, the following corollary
defines the maximum achievable diversity gain by any full
diversity -based space-time code.\\

\begin{corollary}[Maximum diversity]
\label{C2} The maximum diversity gain is the same for both
correlated and uncorrelated spatial fading channels and is given
by:
\begin{IEEEeqnarray}{rCl}
\label{eq-dmax}
\nonumber
   \hat{d}^{max}(\eta) & = & \lim_{r\rightarrow
    0}\hat{d}^{uncorr}(r,\eta)=\lim_{r\rightarrow
    0}\hat{d}^{corr}(r,\eta)\\
    & = & N_tN_r\left[1-\frac{g\eta}{(1+g\eta)\ln(1+g\eta)}\right].
\end{IEEEeqnarray}
\end{corollary}

\begin{proof}
The proof is presented in \cite{Rezki-vtc06}.
\end{proof}
Corollary \ref{C2} agrees with \cite[Theorem 6]{Nara_IT} even
though our diversity estimates are different from those in
\cite{Nara_IT}. Corollary \ref{C2} also indicates that the
estimated maximum diversity gain is unaffected by the spatial
correlation. This has been previously pointed out for a high-SNR
regime in \cite{Paulraj}. Corollary \ref{C2} is however stronger,
since it holds for all values of $\eta$, and in particular for
$\eta=+\infty$.\ It is to be reminded that in establishing this
last result, we have assumed that $\boldsymbol R_{t}$ has full
rank. However, should $\boldsymbol R_{t}$ be rank-deficient, then
$\boldsymbol H$ would also be rank-deficient and it may be
expected that the maximum diversity gain may be lower than that
given by (\ref{eq-dmax}). Moreover, as suggested by
(\ref{UB-mutu-inf}), all the results given here apply when the
receive spatial correlation is considered separately.
\begin{figure}[t]
  \begin{center}
    \includegraphics[scale=0.29]{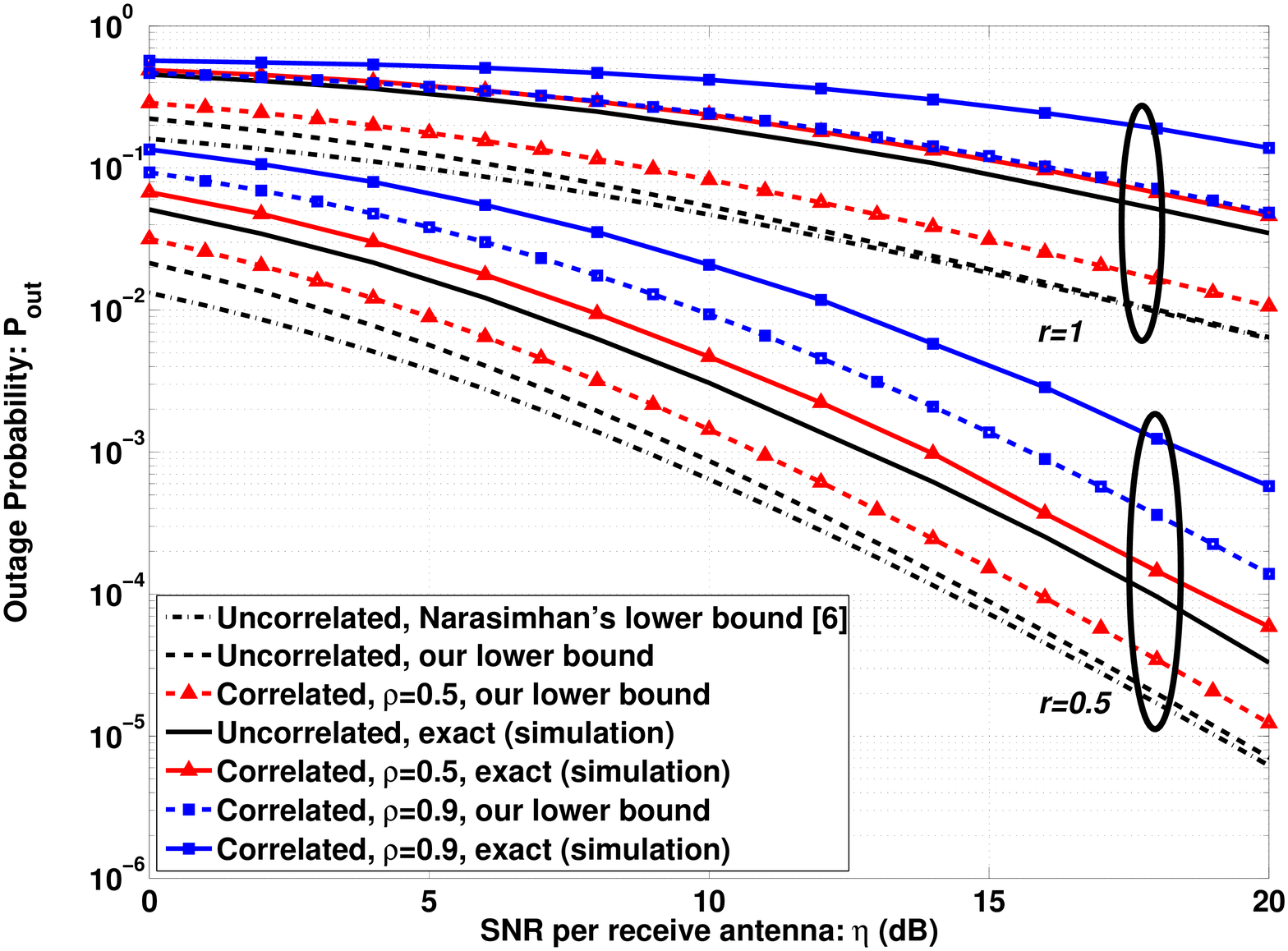}
    \caption{Comparison of lower bounds on $P_{out}$ and the exact simulation in the uncorrelated ($\rho=0$), and correlated ($\rho=0.5,0.9$) spatial fadings, for $N_t=N_r=2$.}
    \label{divmux}
  \end{center}
\end{figure}
\section{Numerical Results}\label{result}
In this section, simulation results for an MEA system with
$N_{t}=N_{r}=2$ are presented. The transmit correlation matrix
$\boldsymbol R_{t}$ is chosen according to a single coefficient
spatial correlation model \cite{Zelst,Paulraj}, i.e., the entry of
$\boldsymbol R_{t}$ at the $i^{th}$ row and the $j^{th}$ column is
$(R_{t})_{i,j} = \rho^{{(i-j)}^2}$. The lower bound on $P_{out}$
given by Theorem \ref{th-LB} is plotted in Fig.\ref{divmux}
together with the exact $P_{out}$, given by simulation, for
$r=0.5$ and $r=1$. Figure \ref{divmux} also shows the lower bound
found by Narasimhan in \cite{Nara_IT} for the uncorrelated case.
As was proven in section \ref{outage}, our bound is tighter
especially at low SNR values. More importantly, our lower bounds
follow the same shape as the exact curves and the gap between the
exact $P_{out}$ and the lower bound is independent of the
correlation coefficients, regardless of the SNR values. For
comparison in the transmit spatial correlation case, we have
plotted in Fig. \ref{divmux2} the exact $P_{out}$, our lower bound
given by (\ref{eq-low-bound2}) and the lower bound in \cite{
Nara_IT}, for $r=0.5$ and $r=1$. All curves in Fig. \ref{divmux2}
have been obtained using the same transmit spatial correlation
matrix in \cite{ Nara_IT}. As can be seen in Fig. \ref{divmux2},
our lower bound is again slightly tighter than Narasimhan's lower
bound at low SNR. Beyond SNR=30 dB, the lower bound curves are
exactly the same.
\begin{figure}[t]
  \begin{center}
    \includegraphics[scale=0.29]{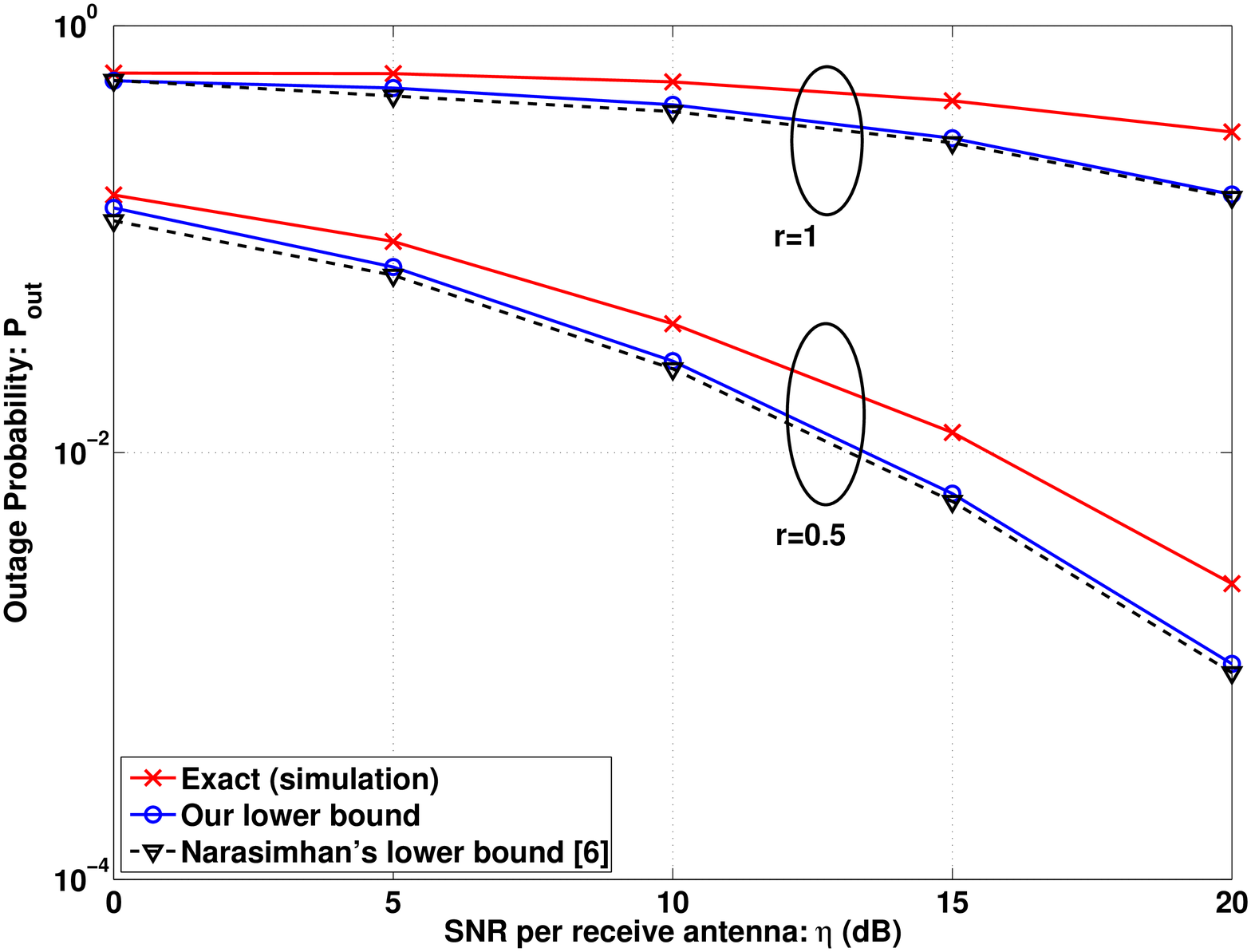}
    \caption{Comparison of our lower bound and Narasimhan's lower bound, for the transmit spatial correlation given by (49) in \cite{ Nara_IT} and for $N_t=N_r=2$.}
    \label{divmux2}
  \end{center}
\end{figure}
The exact diversity gain, obtained by Monte-Carlo simulations
using (\ref{d}) and an estimated diversity gain computed using
Corollary \ref{C1} are plotted in Fig.\ref{d3} for an SNR=15 dB.
Figure \ref{d3} indicates that the DMT estimate is a good fit to
the exact simulation tradeoff curve. Therefore, the estimated
diversity can be used to obtain an insight on the DMT over
spatially correlated and uncorrelated channels while avoiding time
consuming simulations. Interestingly, it can be noticed that with
a correlation coefficient $\rho=0.5$, the diversity gain is only
slightly degraded and one may expect to achieve a quasi-equal
uncorrelated diversity gain as shown in Fig.\ref{d3}. However, the
diversity is substantially degraded when $\rho=0.9$. For example,
as illustrated in Fig.\ref{d2}, an MEA system operating at $r \geq
0.8$ and an SNR of 5 dB in a moderately correlated channel
($\rho=0.5$), achieves a better diversity gain than a system
operating at the same $r$ and an SNR of 10 dB in a highly
correlated channel ($\rho=0.9$). These observations are confirmed
with the exact diversity curves.
\begin{figure}[t]
  \begin{center}
    \includegraphics[scale=0.29]{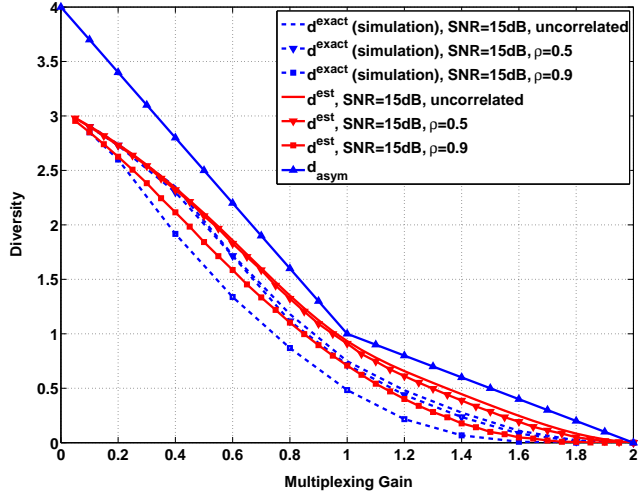}
    \caption{The impact of spatial correlation on the diversity estimates for $N_t=N_r=2$ and SNR=15 dB.}
    \label{d3}
  \end{center}
\end{figure}
In Fig. \ref{d1}, we have plotted the relative diversity-estimate
gain, defined as $\frac{\hat{d}^{corr}}{\hat{d}^{uncorr}}$ for
different SNR values. As predicted by Lemma
\ref{lemmma-asymptoic-behavior}, the relative diversity-estimate
gain converges toward 1 as $\eta \rightarrow \infty$ regardless of
the multiplexing gain values. However, the convergence would be
faster for small values of $r$. Finally, as predicted by Theorem
\ref{th-equality-asymp}, Fig. \ref{d0} illustrates the convergence
of the uncorrelated diversity estimate to the asymptotic DMT as
$\eta \rightarrow \infty$.
\begin{figure}[t]
  \begin{center}
    \includegraphics[scale=0.29]{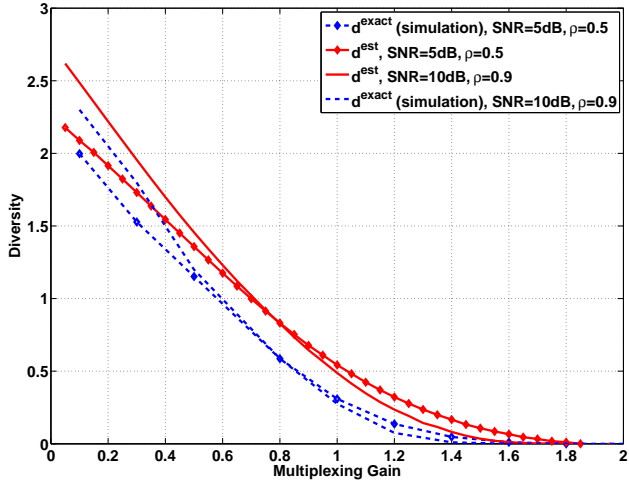}
    \caption{Comparison of our diversity estimate and the exact (simulation) for $\rho=0.5$, $SNR=5$ dB and $\rho=0.9$, $SNR=10$ dB, when $N_t=N_r=2$.}
    \label{d2}
  \end{center}
\end{figure}
\section{Conclusion}\label{conclusion}
In this paper, we have addressed the finite-SNR
diversity-multiplexing tradeoff over spatially correlated MEA
channels. We first derived lower bounds on the outage probability
for spatially correlated and uncorrelated MEA channels. Then,
using these bounds, estimates of the corresponding DMT were
determined. The diversity estimates provide an insight on the
finite-SNR DMT of MEA systems. Furthermore, extensions to the
asymptotic behavior of the diversity-estimate gain, as either the
SNR goes to infinity or the multiplexing gain tends toward zero
have been derived. The asymptotic behavior provides some
guidelines for the design of diversity-oriented space-time codes.
More interestingly, this asymptotic analysis reveals that our
framework includes well-known results about the asymptotic DMT for
correlated and uncorrelated channels. Hence, the framework
presented here can be seen as a generalization of the asymptotic
DMT. Finally, it is worth mentioning that although we have focused
on the transmit spatial correlation, we showed that all the
results still hold when the receive spatial correlation is
considered instead.
\begin{figure}[t]
  \begin{center}
    \includegraphics[scale=0.29]{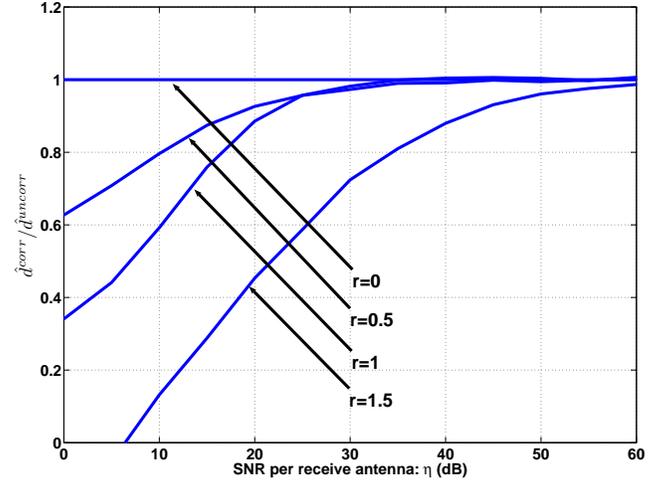}
    \caption{Relative diversity-estimate gain versus SNR for different multiplexing gains $r$.}
    \label{d1}
  \end{center}
\end{figure}
\begin{figure}[t]
  \begin{center}
    \includegraphics[scale=0.29]{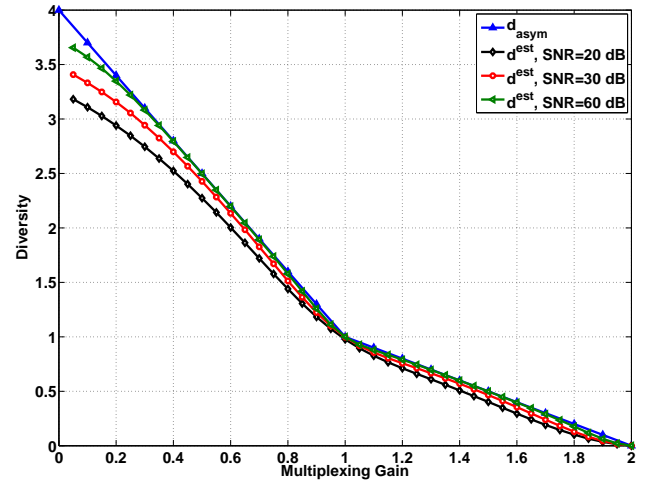}
    \caption{Convergence of the uncorrelated diversity estimate $\hat{d}^{uncorr}(r,\eta)$ to the asymptotic diversity $d_{asym}$, as $SNR \rightarrow \infty$.}
    \label{d0}
  \end{center}
\end{figure}

\appendices
\section{Proof of Lemma \ref{lemmma-asymptoic-behavior}}\label{app-1}
First, note that as $\eta \rightarrow \infty$, we have:
\begin{equation}\label{xil-asymp}
   \xi_{l} \approx N_{t}g^{b_{l}}\eta^{b_{l}-1}.
\end{equation}
Let us define $  J_{l}(\xi_{l}) = \frac{\xi_{l}^{N_t + N_r
-2l}e^{-\xi_l}/(N_t+N_r
 -2l)!}{{\Gamma_{inc}(\xi_l,N_r + N_t -2l+1)}}$, and $K_{l}(\xi_{l}) = \left((1+g\eta)^{b_l}-b_l g \eta (1+g\eta)^{b_l
 -1}-1\right)$, for
$l=1,\ldots,t$. To prove Lemma \ref{lemmma-asymptoic-behavior}, it
suffices to prove that:
\begin{equation}
\label{eq-nec-con}
\lim_{\eta \rightarrow \infty}J_{l}(\xi_{l})K_{l}(\xi_{l}) =
\lim_{\eta \rightarrow
\infty}\frac{Q_{l}(\xi_{l})}{P_{l}(\xi_{l})}K_{l}(\xi_{l}).
\end{equation}
Indeed, to prove (\ref{eq-nec-con}) for each $l=1,\ldots,t$, we
distinguish three cases:
\begin{itemize}
    \item $b_{l} > 1$:
In this case, $\xi_{l}\rightarrow \infty$. Since
$\Gamma_{inc}(\xi_l,N_r + N_t -2l+1) \rightarrow 1$ as $\xi_{l}
\rightarrow \infty$, then $J_{l}(\xi_{l})K_{l}(\xi_{l})
\rightarrow 0$, and so does the term
$\frac{Q_{l}(\xi_{l})}{P_{l}(\xi_{l})}K_{l}(\xi_{l})$.
    \item $b_{l} < 1$:
In this case, $\xi_{l}\rightarrow 0$. Note that
$J_{l}(\xi_{l})=\frac{f'(\xi_{l})}{f(\xi_{l})}$, where $
f(\xi_{l})=\Gamma_{inc}(\xi_l,N_r + N_t -2l+1)$, and $f'(\xi_{l})$
denotes the derivative of $f(\xi_{l})$. Using Taylor expansion of
$f(\xi_{l})$ and $f'(\xi_{l})$ around $0$, we obtain:
\begin{IEEEeqnarray}{rCl}
\nonumber
    J_{l}(\xi_{l}) & \approx & \frac{1/(N_{t}+N_{r}-2l)! \xi_{l}^{N_{t}+N_{r}-2l}}{1/(N_{t}+N_{r}-2l+1)!
    \xi_{l}^{N_{t}+N_{r}-2l+1}}\\
    & = & \frac{N_{t}+N_{r}-2l+1}{\xi_{l}}.
    \label{Jl-approx}
\end{IEEEeqnarray}

On the other hand, $Q_{l}(\xi_{l})=f_{\Delta_{l}}(\xi_{l})$. Since
the $n^{th}$ derivative of $f_{\Delta_{l}}(x)$ can also be
expressed by $f_{\Delta_{l}}^{(n)}(x)= \frac{1}{2 \pi
j}\int_{-j\infty}^{+j\infty}
    (-jv)^{n}e^{-jvx}
    \Psi_{\Delta_{l}}(jv)d(jv)$, $n\geq0$, then it can be shown, using the Residue Theorem, that
$f_{\Delta_{l}}^{(n)}(0)=0$ for $n \leq N_{t}+N_{r}-2l-1$. This is
because $N_{t}+N_{r}-2l+1$ is the degree of
$\Psi_{\Delta_{l}}(jv)$'s denominator. Observing again that
$Q_{l}(\xi_{l})$ is the derivative of $P_{l}(\xi_{l})$ and using
Taylor expansion of $Q_{l}(\xi_{l})$ and $P_{l}(\xi_{l})$ to the
$(N_{t}+N_{r}-2l)^{th}$ term, we also find that
$\frac{Q_{l}(\xi_{l})}{P_{l}(\xi_{l})} \approx
\frac{N_{t}+N_{r}-2l+1}{\xi_{l}}$.
    \item $b_{l} = 1$:
In this case, $\xi_{l}=g N_{t}$. Since $J_{l}(\xi_{l})$ and
$\frac{Q_{l}(\xi_{l})}{P_{l}(\xi_{l})}$ are finite, and
$K_{l}(\xi_{l})=0$, (\ref{eq-nec-con}) still holds.
\end{itemize}

\section{Proof of Theorem \ref{th-equality-asymp}}\label{app-2}
To prove Theorem \ref{th-equality-asymp}, it is sufficient to
prove that $ \lim_{\eta \rightarrow
\infty}\hat{d}^{uncorr}(r,\eta)=d_{asym}$, since the other
equality follows from Lemma \ref{lemmma-asymptoic-behavior}. In
Appendix \ref{app-1}, it was shown that, if $b_{l} \geq 1$,
$l=1,\ldots,t$, then $J_{l}(\xi_{l})K_{l}(\xi_{l}) \rightarrow 0$
as $\eta \rightarrow \infty$. Thus, only the case $b_{l} < 1$ is
of interest. Moreover, using (\ref{xil-asymp}), (\ref{Jl-approx})
and the fact that $K_{l}(\xi_{l})$ can be approximated by
$K_{l}(\xi_{l}) \approx g^{b_{l}}(1-b_{l}) \eta^{b_{l}}$, we
obtain that $J_{l}(\xi_{l})K_{l}(\xi_{l}) \approx
    \frac{N_{t}+N_{r}-2l+1}{N_{t}}(1-b_{l}) \eta$. Hence, $\hat{d}^{uncorr}(r,\eta)$ given by (\ref{d-est-uncorr}) can be
written as:
\begin{IEEEeqnarray}{rCl}
\nonumber
    \hat{d}^{uncorr}(r,\infty) & = & \sum_{l=1,b_{l} <
    1}^{t}(N_{t}+N_{r}-2l+1)(1-b_{l}) \\
    \nonumber
    & = & \sum_{l=1}^{t}(N_{t}+N_{r}-2l+1)(1-b_{l})^{+}\\
    & = & \sum_{l=1}^{t}(N_{t}+N_{r}-2l+1)\alpha_{l},
                  \label{bl-asymp}
\end{IEEEeqnarray}
where $\alpha_{l}=(1-b_{l})^{+}$ and $(x)^{+}=\max{(x,0)}$ for
each real number $x$. Next, we show that the $\alpha_{l}$'s that
satisfy (\ref{bl-asymp}) are exactly the coefficients leading to
the asymptotic DMT given in \cite{ Zheng03}. First, recall that
$b=(b_{1},...,b_{t}) \in \mathcal{A}=\{(b_{1},...,b_{t}) \in
\mathcal{R}^{t+} \mid \sum_{l=1}^{t}b_{l}=r\}$, and $b$ maximizes
the lower bound (\ref{eq-low-bound}).
\begin{equation}\label{eq-max}
    \max_{b \in \mathcal{A}}\prod_{l=1}^{t}\Gamma_{inc}(\xi_l,N_{r}+N_{t}-2l+1).
\end{equation}\
As $\eta \rightarrow \infty$,
$\Gamma_{inc}(\xi_l,N_{r}+N_{t}-2l+1)$ is independent of $b_{l}$,
for $b_{l} \geq 1$. This is because when $b_{l}=1$ then
$\xi_{l}=gN_{t}$, and when $b_{l} > 1$ then $\xi_{l} \rightarrow
\infty$ and $\Gamma_{inc}(\xi_l,N_{r}+N_{t}-2l+1) \rightarrow 1$.
Indeed, if we let $\kappa$ be the number of coefficients $b_{l} <
1$, the maximization problem (\ref{eq-max}) reduces to:
\begin{equation}\label{eq-max2}
    K \cdot \max_{b \in
    \mathcal{B}}\prod_{l=1}^{\kappa}\Gamma_{inc}(\xi_l,N_{r}+N_{t}-2l+1),
\end{equation}
where $K$ is a constant factor and $\mathcal{B}$ is given by:
\begin{equation*}\label{}
    \mathcal{B}=\{(b_{1},...,b_{\kappa}) \in \mathcal{R}^{\kappa+}
\mid \sum_{l=1}^{\kappa}b_{l} \leq r \}.
\end{equation*}
Maximization (\ref{eq-max2}) involves only $b_{l} < 1$, for which
$\xi_{l} \rightarrow 0$ as $\eta \rightarrow \infty$. Using
(\ref{xil-asymp}) and the fact that around zero,
$\Gamma_{inc}(x,m)$ can be approximated by $\Gamma_{inc}(x,m)
\approx \frac{x^{m}}{m!}$, we have:
$\Gamma_{inc}(\xi_l,N_{r}+N_{t}-2l+1) \approx
    C(g \eta)^{(b_{l}-1)(N_{r}+N_{t}-2l+1)}$, where $C=\frac{g N_{t}}{(N_{t}+N_{r}-2l+1)!}$ is a constant
indépendant of $b_{l}$, $l=1,\ldots,t$. Then, (\ref{eq-max2}) is
equivalent to:
\begin{equation*}\label{}
    \min_{b \in
    \mathcal{B}}\sum_{l=1}^{\kappa}(1-b_{l})(N_{r}+N_{t}-2l+1).
\end{equation*}
Thus, the asymptotic diversity estimate given by (\ref{bl-asymp})
can be expressed as:
\begin{equation*}\label{}
     \hat{d}^{uncorr}(r,\infty)= \min_{\alpha
\in \mathcal{A'}}\sum_{l=1}^{t} (N_{r}+N_{t}-2l+1) \alpha_{l},
\end{equation*}
where $\alpha=(\alpha_{1},...,\alpha_{t})$, and $\mathcal{A'}$ is
given by:
\begin{equation*}\label{}
    \mathcal{A'}=\{(\alpha_{1},...,\alpha_{t}) \in \mathcal{R}^{t+} \mid
\sum_{l=1}^{t}(1-\alpha_{l})^{+} \leq r \},
\end{equation*}
which is exactly the asymptotic DMT $d_{asym}$ established in
\cite{ Zheng03}. This completes the proof of Theorem
\ref{th-equality-asymp}.



\begin{thebibliography}{10}
\providecommand{\url}[1]{#1} \csname url@samestyle\endcsname
\providecommand{\newblock}{\relax}
\providecommand{\bibinfo}[2]{#2}
\providecommand{\BIBentrySTDinterwordspacing}{\spaceskip=0pt\relax}
\providecommand{\BIBentryALTinterwordstretchfactor}{4}
\providecommand{\BIBentryALTinterwordspacing}{\spaceskip=\fontdimen2\font
plus \BIBentryALTinterwordstretchfactor\fontdimen3\font minus
  \fontdimen4\font\relax}
\providecommand{\BIBforeignlanguage}[2]{{%
\expandafter\ifx\csname l@#1\endcsname\relax
\typeout{** WARNING: IEEEtran.bst: No hyphenation pattern has been}%
\typeout{** loaded for the language `#1'. Using the pattern for}%
\typeout{** the default language instead.}%
\else \language=\csname l@#1\endcsname \fi #2}}
\providecommand{\BIBdecl}{\relax} \BIBdecl

\bibitem{Tarokh98}
V.~Tarokh, N.~Seshadri, and A.~Calderbank, ``Space-time codes for
high data
  rate wireless communication: performance criterion and code construction,''
  \emph{{IEEE} Trans. Inform. Theory}, vol.~44, no.~2, pp. 744--765, 1998.

\bibitem{Foschini96}
G.~Foschini, ``Layered space time architecture for wireless
communication in a
  fading environment when using multi-element antennas,'' \emph{Bell Systems
  Technical Journal}, vol.~1, pp. 41--59, Autumn 1996.

\bibitem{Zheng03}
L.~Zheng and D.~N.~C. Tse, ``Diversity and multiplexing: A
fundamental tradeoff
  in multiple-antenna channels,'' \emph{{IEEE} Trans. Inform. Theory}, vol.~49,
  no.~5, pp. 1073--1096, May 2003.

\bibitem{Nara}
R.~Narasimhan, ``{Finite-SNR} diversity performance of
rate-adaptive {MIMO}
  systems,'' in \emph{Proc., IEEE Global Telecomm. Conf.}, 28 Nov.-2 Dec. 2005,
  pp. 1461-- 1465.

\bibitem{Da}
D.-S. Shiu, G.~Foschini, M.~Gans, and J.~Kahn, ``Fading
correlation and its
  effect on the capacity of multielement antenna systems,'' \emph{{IEEE} Trans.
  Commun.}, vol.~48, no.~3, pp. 502--513, Mar. 2000.

\bibitem{Nara_IT}
R.~Narasimhan, ``{Finite-SNR} diversity-multiplexing tradeoff for
correlated
  {Rayleigh} and {Rician} {MIMO} channels,'' \emph{{IEEE} Trans. Inform.
  Theory}, vol.~52, no.~9, pp. 3965-- 3979, Sep. 2006.

\bibitem{Rezki-vtc06}
Z.~Rezki, B.~Cotruta, D.~Haccoun, and F.~Gagnon, ``Finite
  diversity-multiplexing tradeoff over a spatially correlated channel,'' in
  \emph{Proc., IEEE Veh. Technol. Conf.}, {Montr\'eal}, Qc, Canada, 25-28 Sep.
  2006.

\bibitem{Chang}
W.~Chang, S.-Y. Chung, and Y.~H. Lee, ``Diversity-multiplexing
tradeoff in
  rank-deficient and spatially correlated {MIMO} channels,'' in
  \emph{International Symposium on Information Theory}, Seatle, Washington,
  9-14 Jul. 2006.

\bibitem{Paulraj}
A.~Paulraj, R.~Nabar, and D.~Gore, \emph{Introduction to
Space-Time Wireless
  Communications}.\hskip 1em plus 0.5em minus 0.4em\relax Cambridge University
  Press, 2003.

\bibitem{Telatar}
I.~E. Telatar, ``Capacity of multi-antenna gaussian channels,''
\emph{Europeen
  Trans. On Communication}, vol.~10, no.~6, pp. 585--5595, Nov. 1999.

\bibitem{Zelst}
A.~V. Zelst and J.~Hammerschmidt, ``A single coefficient spatial
correlation
  model for multiple-intput multiple-output {(MIMO)} radio channels,'' in
  \emph{27th General Assembly of the International Union of Radio Science
  (URSI), Maastricht, the Netherlands}, Aug.2002., pp. 1461-- 1465.

\end{thebibliography}
\end{document}